\RequirePackage{fix-cm}
\documentclass{svjour3}                     

\usepackage{framed}
\usepackage{amssymb}
\usepackage{amsmath}
\usepackage{algorithm, algpseudocode}
\usepackage{appendix}

\newcommand{\old}[1]{}
\newtheorem{observation}{Observation}
\smartqed  
\usepackage{graphicx}
\begin{document}

\title{Broadcast and minimum spanning tree with $o(m)$ messages in the asynchronous CONGEST model \thanks{Funded with an NSERC grant.
}
}
\titlerunning{Broadcast and MST with $o(m)$ messages}        

\author{Ali Mashreghi         \and
        Valerie King 
}


\institute{A. Mashreghi \at
              Department of Computer Science, University of Victoria, BC, Canada\\
              \email{ali.mashreghi87@gmail.com}           
           \and
           V. King \at
              Department of Computer Science, University of Victoria, BC, Canada
}

\date{Received: date / Accepted: date}
\maketitle

\begin{abstract}
We provide the first asynchronous distributed algorithms to compute broadcast and minimum spanning tree with $o(m)$ bits of communication, in a sufficiently dense graph with $n$ nodes and $m$ edges. For decades, it was believed that $\Omega(m)$ bits of communication are required for any algorithm that constructs a broadcast tree. In 2015, King, Kutten and Thorup showed that in the KT1 model where nodes have initial knowledge of their neighbors' identities it is possible to construct MST in $\tilde{O}(n)$ messages in the synchronous CONGEST model. In the CONGEST model messages are of size $O(\log n)$. However, no algorithm with $o(m)$ messages were known for the asynchronous case. Here, we provide an algorithm that uses $O(n^{3/2} \log^{3/2} n)$ messages to find MST in the asynchronous CONGEST model. Our algorithm is randomized Monte Carlo and outputs MST with high probability. We will provide an algorithm for computing a spanning tree with $O(n^{3/2} \log^{3/2} n)$ messages.  Given a spanning tree, we can compute MST with $\tilde{O}(n)$ messages. 
\end{abstract}

\section{Introduction}
We consider a distributed network as an undirected graph with $n$ nodes and $m$ edges, and the problem of finding a spanning tree and a minimum spanning tree (MST) with efficient communication. That is, we require that every node in the graph learns exactly the subset of its incident edges which are in the spanning tree or MST, resp.  
 A spanning tree enables a message to be broadcast from one node to all other nodes with only $n-1$ edge traversals. In a sensor or ad hoc network where the weight of a link between nodes reflects the amount of energy required to transmit a message along the link \cite{kutten2015complexity}, the minimum spanning tree (MST) provides an energy efficient means of broadcasting. The problem of finding a spanning tree in a network has been studied for more than three decades, since it is the building block of many other fundamental problems such as \emph{counting}, \emph{leader election}, and \emph{deadlock resolution} \cite{awerbuch1987optimal}. 

 A spanning tree can be constructed by a simple breadth-first search from a single node using $m$ bits of communication.  The tightness of this communication bound was  a ``folk theorem'',  according to Awerbuch, Goldreich, Peleg and Vainish \cite{awerbuch1990trade}. Their 1990 paper defined the KT1 model where nodes have unique IDs and know only their neighbors. It showed, for a limited class of algorithms, a lower bound of $\Omega(m)$ messages in a synchronous KT1 network. In 2015, Kutten et al. \cite{kutten2015complexity} proved a lower bound for general randomized algorithms with $O(\log n)$ bit messages, in the KT0 model, where nodes do not know their neighbors. In 2015, King, Kutten, and Thorup  gave the first distributed algorithm (``KKT'') with $o(m)$ communication to build a broadcast tree and MST in the KT1 model. They devised Monte Carlo algorithms in the synchronous KT1 model  with $\tilde{O}(n)$ communication \cite{king2015construction}. This paper and a followup paper \cite{mashreghi2017time} left open the problem of whether a $o(m)$ bit communication algorithm in the asynchronous model was possible, for either the spanning tree or MST problem, when nodes know their neighbors' IDs.

In an asynchronous network, there is no global clock. All processors may wake up at the start and send messages, but further actions by a node are event-driven, i.e., in response to messages received.   The pioneer work of Gallager, Humblet, and Spira \cite{gallager1983distributed} (``GHS'') presented an asynchronous protocol for finding the MST in the CONGEST model, where messages are of size $O(\log n)$.  GHS requires $O(m + n\log n)$ messages and $O(n \log n)$ time if all nodes are awakened simultaneously. 
Afterwards, researchers worked on improving the time complexity of MST algorithms in the CONGEST model
but the message complexity remained $\Omega{(m)}$. 
In this paper, we provide the first algorithm in the KT1 model which uses $o(m)$ bits of communication for finding a spanning tree in an asynchronous network, specifically we show the following:

\begin{theorem} \label{theorem}
Given any network of $n$ nodes where all nodes awake at the start, a  spanning tree and a minimum spanning tree can be built with $O(n^{3/2} \log^{3/2} n)$ messages  in the asynchronous  KT1 CONGEST model, with high probability. 
\end{theorem}

\subsection{Techniques:}  ~~Many distributed algorithms to find an MST use the Boruvka method:  Starting from the set of isolated nodes, a forest of  edge disjoint rooted trees which are subtrees of the MST are maintained. The algorithms runs in phases:  In a phase, in parallel, each tree $A$ finds a minimum weight {\it outgoing edge}, that is, one with exactly one endpoint in $A$ and its other endpoint in some other tree B. Then the outgoing edge is inserted to create the ``merged'' tree containing the nodes of $A$ and $B$. In what seems an inherently synchronous process, every tree (or a constant fraction of the trees) participates in some merge, the number of trees is reduced by a constant factor per phase, and $O(\log n)$ phases suffice to form a single tree. \cite{gallager1983distributed,awerbuch1987optimal,king2015construction,mashreghi2017time}.

 The KKT paper introduced procedures $FindAny$ and $FindMin$ which can find any or the minimum outgoing edge leaving the tree, respectively. These require ${O}(|T|)$ messages and $\tilde{O}(|T|)$, resp., where $|T|$ is the number of nodes in the tree $T$ or a total of $\tilde{O}(n)$ per phase. As this is done synchronously in KKT, only $O(\log n)$ phases are needed, for a total number of only ${O}(n \log n)$ messages to build a spanning tree.

While $FindAny$ and $FindMin$ are asynchronous procedures, the  Boruvka approach of  \cite{king2015construction} does not seem to work in an asynchronous model  with $o(m)$ messages, as it does not seem possible to prevent only one tree from growing, one node at a time, while the other nodes are delayed, for a cost of $O(n^2)$ messages. The asynchronous GHS also uses $O(\log n)$ phases to merge trees in parallel,  but it is able to synchronize the growth of the trees by assigning a \textit{rank} to each tree. A tree which finds a minimum outgoing edge waits to merge until the tree it is merging with is of equal or higher rank. 
 The GHS algorithm subtly avoids traversing the whole tree until a minimum weight outgoing edge to an appropriately ranked tree is found. This method seems to require communication over all edges in the worst case.  

Asynchrony precludes approaches that can be used in the synchronous model. For example, in the synchronous model, if nodes of low degree send messages to all their neighbors, in one round all nodes learn which of their neighbors do not have low degree, and therefore they can construct the subgraph of higher degree nodes. In the asynchronous model, a node, not hearing from its neighbor, does not know when to conclude that its neighbor is  of higher degree.

 The technique for building a spanning tree in our paper is very different from the technique in \cite{king2015construction} or \cite{gallager1983distributed}.
We grow one tree $T$ rooted at one preselected {\it leader} in phases.  (If there is no preselected leader, then this may be done from a small number of randomly self-selected nodes.)
 Initially, each node selects itself with probability $1/\sqrt{n \log n}$ as a \emph{star node}. (We use $\log n$ to denote $\log_2 n$.) This technique is inspired from \cite{Elkin:2017:DES:3055399.3055452}, and provides a useful property that every node whose degree is at least $\sqrt{n} \log^{3/2} n$ is adjacent to a star node with high probability. Initially, star nodes  (and low-degree nodes) send out messages to all of their neighbors. Each high-degree node which joins $T$ waits until it hears from a star node and then invites it to join $T$.  In addition, when low-degree and star nodes join $T$, they invite their neighbors to link to $T$ via their incident edges.  Therefore, with high probability, the following invariant for $T$ is maintained as $T$ grows: \\

\noindent
{\bf Invariant:} $T$ includes all neighbors of any star or low-degree node in $T$, as well.  Each high-degree node in $T$ is adjacent to a star node in $T$. \\

 The challenge is for high-degree nodes in $T$ to find neighbors outside $T$. If in each phase, an outgoing edge from a high-degree node in $T$ to a high-degree node $x$ (not in $T$) is found and $x$ is invited to join $T$, then $x$'s adjacent star node (which must lie outside $T$ by the Invariant) is also found and invited to join. As the number of star nodes is  $O(\sqrt{ n} / \log^{1/2} n)$, this number also bounds the number of such phases. The difficulty is that there is no obvious way to find an outgoing edge to a high degree node because, as mentioned above, in an asynchronous network, a high degree node has no apparent way to determine if its neighbor has high degree without receiving a message from its neighbor.
 
Instead, we relax our requirement for a phase. With each phase either \textbf{(A)} A high-degree node (and star node) is added to $T$ or \textbf{(B)} $T$ is expanded so that the number of outgoing edges to low-degree nodes is reduced by a constant factor. As there are no more than $O(\sqrt{n} / \log^{1/2} n)$ phases of type \textbf{A}  and no more than $O(\log n)$ phases of type \textbf{B} between each type \textbf{A} phase, there are a total of $O(\sqrt{n} \log^{1/2} n)$ phases before all nodes are in $T$.
The key idea for implementing a phase of type \textbf{B} is that the tree $T$ waits until its nodes have heard enough messages passed by low-degree nodes over outgoing edges before initiating an expansion.
 The efficient implementation of a phase, which uses only $O(n\log n)$ messages, requires a number of tools which are described in the preliminaries section. 

Once a spanning tree is built, we use it as a communication network while finding the MST. This enables us to ``synchronize'' a modified GHS which uses $FindMin$ for finding minimum outgoing edges, using a total of $\tilde{O}(n)$ messages. \par

\subsection{Related work:}~
The Awerbuch, Goldreich, Peleg and Vainish \cite{awerbuch1990trade} lower bound on the number of messages holds only for (randomized) algorithms  where messages may contain  a constant number of IDs, and IDs are processed by comparison only and for general deterministic algorithms, where ID's are drawn from a very large size universe. 

Time to build an MST in the CONGEST model has been explored in several papers. Algorithms include, in the asynchronous KT0 model,
\cite{gallager1983distributed,awerbuch1987optimal,faloutsos1995optimal,singh1995highly}, and in the synchronous KT0 model, \cite{kutten1995fast,garay1998sublinear,elkin2004faster,Khan:2006:FDA:2136050.2136076}. Recently, in the synchronous KT0 model,  Pandurangan gave a \cite{pandurangan2017time}  $\tilde{O}(D + \sqrt{n})$ time and  $\tilde{O}(m)$ message randomized algorithm, which Elkin improved by logarithmic factors with a deterministic algorithm  \cite{elkin2017simple}. The time complexity to compute spanning tree in the algorithm of  \cite{king2015construction} is $O(n \log n)$ which was improved to $O(n)$ in \cite{mashreghi2017time}.

Lower bounds on time for approximating the minimum spanning tree has been proved in the synchronous KT0 model In \cite{elkin2006unconditional,sarma2012distributed} . Kutten et al. \cite{kutten2015complexity} show an $\Omega(m)$ lower bound on message complexity for randomized general algorithms in the KT0 model. \par 
$FindAny$ and $FindMin$ which appear in the KKT algorithms build on ideas for sequential dynamic connectivity in \cite{kapron2013dynamic}. A sequential dynamic $ApproxCut$ also appeared in that paper \cite{kapron2013dynamic}. Solutions to the sequential linear sketching problem for connectivity \cite{ahn2012graph} share similar techniques but require a more complex step to verify when a candidate edge name is an actual edge in the graph, as the edges names are no longer accessible once the sketch is made (See Subsection \ref{s:findany}).

The  \textit{threshold detection} problem was introduced by Emek and Korman \cite{emek2010efficient}. It assumes that there is a rooted spanning tree $T$ where events arrive online at $T$'s nodes. Given some threshold $k$, a termination signal is broadcast by the root if and only if the number of events exceeds $k$.  We use a naive solution of a simple version of the problem here.

A synchronizer, introduced by Awerbuch \cite{awerbuch1985complexity} and studied in \cite{awerbuch1990network,awerbuch2007time,peleg1987optimal,elkin2008synchronizers}, is a general technique for simulating a synchronous algorithm on an asynchronous network using communications along a spanning tree.  To do this, the spanning tree must be built first. Using a general synchronizer imposes an overhead of messages that affect \emph{every single step} of the synchronous algorithm that one wants to simulate, and would require more communication than our special purpose method of using our spanning tree to synchronize the modified GHS.

\subsection{Organization:} ~Section \ref{modelsec} describes the model. Section \ref{STconstruction} gives the spanning tree algorithm for the case of a connected network and a single known leader. Finally, Section \ref{MSTconstruction} provides the MST algorithm. Section \ref{disconCase} provides the algorithm for computing a minimum spanning \emph{forest} in disconnected graphs.

\section{Preliminaries}
\label{modelsec}

\subsection{Model:}
 ~~Let  $c \geq 1$ be any constant. The communications network is the undirected graph $ G=(V,E)$ over which a spanning tree or MST will be found.  Edge weights are integers in $[1, n^c]$.  IDs are assigned uniquely by the adversary from $[1,n^c]$. All nodes have knowledge of $c$ and $n$ which is an upper bound on $|V|$ (number of nodes in the network) within a \emph{constant} factor. All nodes know their own ID along with the ID of their neighbors (KT1 model) and the weights of their incident edges. Nodes have no other information about the network. e.g., they do not know $|E|$ or the maximum degree of the nodes in the network.  Nodes can only send direct messages to the nodes that are adjacent to them in the network.  If the edge weights are not unique they can be made unique by appending the ID of the endpoints to its weight, so that the MST is unique.  Nodes can only send direct messages to the nodes that are adjacent to them in the network. Our algorithm is described in the CONGEST model in which each message has size $O(\log n)$.  Its time is trivially bounded by the total number of messages.  The  KT1  CONGEST model has been referred to as the ``standard  model''\cite{awerbuch1990trade}.

Message cost is the sum over all edges of the number of messages sent over each edge during the execution of the algorithm. 
If a message is sent it is eventually received, but the adversary controls the length of the delays and there is no guarantee that messages sent by the same node will be received in the order they are sent. There is no global clock. All nodes awake at the start of the protocol simultaneously. After awaking and possibly sending its initial messages, a processor acts only in response to receiving messages.

We say a network ``finds'' a subgraph  if at the end of the distributed algorithm, every node knows exactly which  of its incident edges in the network are part of the subgraph.
The algorithm here is Monte Carlo, in that it succeeds with probability $1-n^{-c''}$ for any constant $c''$ (``w.h.p.'').

We initially assume there is a special node (called \emph{leader}) at the start and the graph is connected. These assumptions are dropped in the algorithm we provide for disconnected graphs in the full version of the paper.

\subsection{Definitions and Subroutines:}
\label{defs}
~~$T$ is initially a tree containing only the leader node. Thereafter, $T$ is a tree rooted at the leader node.
We use the term  {\it outgoing edge} from  $T$ to mean an edge with exactly one endpoint in $T$.   An outgoing edge is described as if it is directed;  it is {\it from} a node in $T$ and {\it to} a node not in $T$ (the ``external'' endpoint).

The algorithm uses the following subroutines and definitions:
\begin{itemize}

\item
$Broadcast(M)$: Procedure whereby the node $v$ in $T$ sends message $ M $ to its children and its children broadcast to their subtrees. 

\item
$Expand$: A procedure for adding nodes to $T$ and preserving the Invariant after doing so.

\item $ FindAny$:   Returns to the leader an outgoing edge chosen uniformly at random with probability 1/16, or else it returns $\emptyset$. The leader then broadcasts the result.  $FindAny$ requires $O(n)$ messages. We specify $FindAny(E')$ when we mean that the outgoing edge must be an outgoing edge in a particular subset $E' \subseteq E$.

\item $FindMin$:  is similarly defined except the edge is the (unique) minimum cost outgoing edge. This is used only in the minimum spanning tree algorithm. $FindMin$ requires $O(n \log^2 n/\log \log n)$ messages.

\item
$ApproxCut $: A function which w.h.p.
returns an estimate in $[k/32, k]$ where $k$ is the number of outgoing edges from $T$ and $k > c \log n$ for $c$ a constant.
It requires $O(n \log n)$ messages.


$FindAny$ and $FindMin$ are described in \cite{king2015construction} (The $FindAny$ we use is called $\textit{FindAny-C}$ there.) $\textit{FindAny-C}$ was used to find {\it any} outgoing edge in the previous paper. It is not hard to see that the edge found is a random edge from the set of outgoing edges; we use that fact here. The relationships among $FindAny$, $FindMin$ and $ApproxCut$ below are described in the next subsection.
 
\item $Found_L(v)$, $Found_O(v)$: Two lists of edges incident to node $v$, over which $v$ will send invitations to join $T$ the next time $v$ participates in $Expand$. After this, the list is emptied.  Edges are added to $Found_L(v)$ when $v$ receives $\langle \textit{Low-degree} \rangle$ message or the edge is found by the leader by sampling and its external endpoint is low-degree. Otherwise, an edge is added to $Found_O(v)$ when $v$ receives a $\langle \textit{Star} \rangle$ message over an edge or if the edge is found by the leader by sampling and its external endpoint is high-degree. Note that star nodes that are low-degree send both $\langle \textit{Low-degree} \rangle$ and $\langle \textit{Star} \rangle$. This may cause an edge to be in both lists which is handled properly in the algorithm.

\item $\textit{T-neighbor}(v)$: A list of neighbors of $v$ in $T$. This list, except perhaps during the execution of $Expand$, includes all low-degree neighbors of $v$ in $T$. This list is used to exclude from $Found_L(v)$ any non-outgoing edges.

\item
$ThresholdDetection(k)$:  A procedure which is initiated by the leader of $T$.  The nodes in $T$ experience no more than $k<n^2$ events w.h.p. The leader is informed w.h.p. when the number of events experienced by the nodes in $T$ reaches the threshold $k/4$. Here, an event is the receipt of $\langle \textit{Low-degree} \rangle$ over an outgoing edge.  Following the completion of $Expand$, all edges $(u,v)$ in $Found_L(u)$ are events if $v \notin \textit{T-neighbor}(u)$. $O(|T|\log n)$ messages suffice.
\end{itemize}

\textit{Note:} The reason we need the $\textit{T-neighbor}$ data structure is that we assumed that messages may not be received the same order they were sent. In particular, $\langle \textit{Low-degree} \rangle$ messages from neighbors of a node in the fragment tree could be received very late. Therefore, we need to have this data structure so that when the fragment is waiting for events (in $ThresholdDetection$) we do not count these late messages as events. An alternative way to deal with this is to have nodes send back acknowledgment messages if they receive a $\langle \textit{Low-degree} \rangle$  message. A  low-degree node will only send its future messages only if it has received an acknowledgement for its initial $\langle \textit{Low-degree} \rangle$  message. This guarantees that no $\langle \textit{Low-degree} \rangle$  message from a node inside the fragment will contribute to the events.

\subsection{Implementation of $FindAny$, $FindMin$ and $ApproxCut$:} \label{s:findany}

~~We briefly review $FindAny$ in \cite{king2015construction} and explain its connection with $ApproxCut$.  The key insight is that an outgoing edge is incident to exactly one endpoint in $T$ while other edges are incident to zero or two endpoints.  If there were exactly one outgoing edge, the parity of the sum of all degrees in $T$ would be 1, and the parity of bit-wise XOR of the binary representation of the names of all incident edges would be the name of the one outgoing edge. 

To deal with possibility of more than one outgoing edge,  the leader creates an efficient means of sampling edges at different rates: Let $l=\lceil{2\log n }\rceil$. The leader selects and broadcasts one pairwise independent hash function $h:[edge\_names]\rightarrow [1,2^l]$, where $edge\_name$ of an edge is a unique binary string computable by both its endpoints, e.g., $\{x,y\}= x \cdot y$ for $x<y$. Each node $y$ forms the vector $\overrightarrow{h(y)}$ whose $i^{th}$ bit is the parity of its incident edges that hash to $[0,2^i]$, $i = 0, \ldots, l$.   Starting with the leaves, a node in $T$ computes the bitwise XOR of the vectors from its children and itself and then passes this up the tree, until the leader has computed $\overrightarrow{b}=XOR_{y \in T} \overrightarrow{h(y)}$. The key insight implies that for each index $i$, $ \overrightarrow{b_i}$ equals the parity of just the outgoing edges mapped to $[0,2^i]$. Let $min$ be the smallest index $i$ s.t. $\overrightarrow{b_i}=1$. With constant probability,  exactly one edge of the outgoing edges has been mapped to $[1,2^{min}]$. 
The leader broadcasts $min$. Nodes send back up the XOR of the $edge\_names$ of incident edges which are mapped by $h$ to this range. If exactly one outgoing edge has been indeed mapped to that range, the leader will find it by again determining the XOR of the $edge\_names$ sent up. One more broadcast from the leader can be used to verify that this edge exists and is incident to exactly one node in $T$.

Since each edge has the same probability of failing in $[0,2^{min}]$, this procedure gives a randomly selected edge. Note also that the leader can instruct the nodes to exclude certain edges from the XOR, say incident edges of weight greater than some $w$. In this way the leader can binary search for the minimal weight outgoing edge to carry out $FindMin$.  Similarly, the leader can select random edges without replacement.

Observe that if the number of outgoing edges is close to  $2^{j}$, we'd expect $min$ to be $l-j$ with constant probability. Here we introduce distributed asynchronous $ApproxCut$ which uses the sampling technique from $FindMin$ but repeats it $O(\log n)$ times with $O(\log n)$ randomly chosen hash functions. Let $min\_sum$ be the minimum $i$ for which the sum of $ \overrightarrow{b_i}$s exceeds $c \log n$ for some constant $c$. We show $2^{min\_sum}$ approximates the number of outgoing edges within a constant factor from the actual number. $ApproxCut$ pseudocode is given in Algorithm 5.

We show:
\begin{lemma}
With probability $1-1/n^c$,  \emph{ApproxCut} returns an estimate in $[k/32, k]$ where $k$ is the number of outgoing edges and $k > c'\log n$, $c'$ a constant depending on $c$. It uses $O(n \log n)$ messages.
\end{lemma}

The proof is given in Section \ref{s:approxcut}.

\section{Asynchronous ST construction with $o(m)$ messages}
\label{STconstruction}
In this section we explain how to construct a spanning tree when there is a preselected leader and the graph is connected.

Initially, each node selects itself with probability $1/\sqrt{n \log n}$ as a \emph{star node}. Low-degree and star nodes initially send out $\langle \textit{Low-degree} \rangle$ and $\langle \textit{Star} \rangle$ messages to all of their neighbors, respectively. (We will be using the $\langle M \rangle$ notation to show a message with content $M$.) A low-degree node which is a star node sends both types of messages. At any point during the algorithm, if a node $v$ receives a $\langle \textit{Low-degree} \rangle$ or $\langle \textit{Star} \rangle$ message through some edge $e$, it adds $e$ to $Found_L(v)$ or $Found_O(v)$ resp.

The algorithm FindST-Leader runs in phases. Each phase has three parts: {\bf 1)  Expansion} of $T$ over found edges since the previous phase and restoration of the Invariant; {\bf 2) Search}  for an outgoing edge to a high-degree node;  {\bf 3) Wait} until messages to nodes in $T$ have been received over a {\it constant fraction of the outgoing edges} whose external endpoint is low-degree. 

\bigskip
\noindent 
{\bf 1) Expansion:}
Each phase is started with $Expand$. $Expand$ adds to $T$ any nodes which are external endpoints of outgoing edges placed on a $Found$ list of any node in $T$ since the last time that node executed $Expand$. In addition, it restores the Invariant for $T$.  \\

\noindent
{\it Implementation:} Expand is initiated by the leader and broadcast down the tree. When a node $v$ receives $\langle \textit{Expand} \rangle$ message for the first time (it is not in $T$),  it joins $T$ and makes the sender its parent. If it is a high-degree node and is not a star, it has to wait until it receives a $\langle \textit{Star} \rangle$ message over some edge $e$, and then adds $e$ to $Found_O(v)$. It then forwards $\langle \textit{Expand} \rangle$ over the edges in $Found_L(v) $ or $Found_O(v)$ and empties these lists. Otherwise, if it is a low-degree node or a star node, it forwards $\langle \textit{Expand} \rangle$ to \textit{all} of its neighbors. 

On the other hand, if $v$ is already in $T$, it forwards  $\langle \textit{Expand} \rangle$ message to its children in $T$ and along any edges in $Found_L(v)$ or $Found_O(v)$, i.e. outgoing edges which were ``found'' since the previous phase, and empties these lists. All $\langle \textit{Expand} \rangle$ requests received by $v$ are answered, and their sender is  added to $\textit{T-neighbor}(v)$. The procedure ends in a bottom-up way and ensures that each node has heard from all the nodes it sent $\langle \textit{Expand} \rangle$ requests to before it contacts its parent.  

Let  $T^i$ denote $T$ after the execution of $Expand$ in phase $i$. Initially $T^0$ consists of the leader node and as its Found lists contain all its neighbors, after the first execution of $Expand$, if the leader is high-degree, $T_1$ satisfies the invariant. An easy inductive argument on $T^i$ shows:

\begin{observation}
\label{o:expand}
For all $i>0$,  upon completion of $Expand$, all the nodes reachable by edges in the $Found$ lists of any node in $T^{i-1}$ are in $T^i$, and for all $v \in T$, $\textit{T-neighbor}(v)$ contains all the low-degree neighbors of $v$ in $T$.  
\end{observation}

 $Expand$ is called in line \ref{expandLine} of the main algorithm \ref{findMST}. The pseudocode is given in $Expand$ Algorithm  \ref{expandAlg}.
 
\bigskip
 
\noindent
{\bf 2) Search for an outgoing edge to a high degree node:}
A sampling of the outgoing edges without replacement is done using $FindAny$ multiple times. The sampling either (1) finds an outgoing edge to a high degree node, or (2) finds all outgoing edges, or (3) determines  w.h.p. that at least  half the outgoing edges are to low-degree nodes and there are at least 
$2c \log n$ such edges.  If the first two cases occur, the phase ends.\\  

\noindent
{\it Implementation:} Endpoints of sampled edges in $T$ communicate over the outgoing edge to determine if the external endpoint is high-degree. If at least one is, that edge is added to the $Found_O$ list of its endpoint in $T$ and the phase ends. If there are fewer than 
$2\log n$ outgoing edges, all these edges are added to $Found_O$ and the phase ends. If there are no outgoing edges, the algorithm ends.  If all $2\log n$ edges go to low-degree nodes, then the phase continues with Step 3) below. This is implemented in the {\bf while} loop of FindST-Leader.
\\
\\
Throughout this section we will be using the following fact from Chernoff bounds:\\
Assume $X_1, X_2, \ldots, X_T$ are independent Bernoulli trials where each trial's outcome is 1 with probability $0 < p < 1$. Chernoff bounds imply that given constants $c$, $c_1 >1$ and $c_2<1$  there is a constant $c''$ such that if there are $T \geq  c'' \log n$ independent trials, then $Pr(X > c_1 \cdot E[X]) < 1/n^c$ and  $Pr(X < c_2 \cdot E[X]) < 1/n^c$, where $X$ is sum of the $X_1, \ldots, X_T$.\\

We show:

\begin{lemma}\
\label{aftersearch}
After Search, at least one of the following must be true with probability $1-1/n^{c'}$, where $c'$ is a constant depending on $c$: 1) there are fewer than $2c \log n$ outgoing edges and the leader learns them all; 2) an outgoing edge to a high-degree node is found, or 3)
 there are at least  $2c \log n$ outgoing edges and at least half the outgoing edges are to low-degree nodes.
 \end{lemma}

\begin{proof}
Each $FindAny$ has a probability of  1/16 of returning an outgoing edge and if it returns an edge, it is always outgoing. After $48 c \log n$ repetitions without replacement, the expected number of edges returned is $3c \log n$. As these trials are independent, Chernoff bounds imply that at least 2/3 of trials will be successful with probability at least $1-1/n^c$, i.e.,  $2 c \log n$ edges are returned if there are that many, and if there are fewer, all will be returned.

The edges are picked  uniformly at random by independent repetitions of $FindAny$. If more than half the outgoing edges are to high-degree nodes, the probability that all edges returned are to low-degree nodes is $1/2^{2c \log n} < 1/n^{2c}$.

\end{proof}

\bigskip

\noindent
{\bf  3) Wait to hear from outgoing edges to low-degree external nodes:}
This step forces the leader to wait until $T$ has been contacted over a constant fraction of the outgoing edges to (external) low-degree nodes. Note that we do not know how to give a good estimate on the number of low-degree nodes which are neighbors of $T$.  Instead we count outgoing edges. \\

\noindent
{\it Implementation:} 
This step occurs only if the $2c \log n$ randomly sampled outgoing edges all go to low-degree nodes and therefore the number of outgoing edges to low-degree nodes is at least this number. In this case, the leader waits until $T$ has been contacted through a constant fraction of these edges.

If this step occurs, then w.h.p., at least half the outgoing edges go to low-degree nodes. Let $k$ be the number of outgoing edges; $k \geq 2c \log n$. The leader calls $ApproxCut$ to return an estimate $q \in [k/32, k]$ w.h.p. It follows that w.h.p. the number of outgoing edges to low-degree nodes is $k/2$. Let $r=q/2$. Then $r \in [k/64,k/2]$.

The nodes $v \in T$ will eventually receive at least $k/2$ $\langle \textit{Low-degree} \rangle$ messages over the outgoing edges.  Note that these messages must have been received by $v$ after $v$ executed $Expand$ and added to $Found_L(v)$, for otherwise, these would not be outgoing edges. 

The leader initiates a {\it ThresholdDetection} procedure whereby there is an event for a node $v$ for each outgoing edge $v$ has received a 
$\langle \textit{Low-degree} \rangle$ message over since the last time $v$ executed $Expand$. As the {\it ThresholdDetection} procedure is initiated after the leader finishes $Expand$, the $\textit{T-neighbor}(v) $ includes any low-degree neighbor of $v$ that is in $T$. Using $\textit{T-neighbor}(v)$, $v$ can determine which edges in $Found_L(v)$ are outgoing. 

Each event experienced by a node causes it to flip a coin with probability $\min\{c \log n/r, 1\}$. If the coin is heads, then a trigger message labelled with the phase number is sent up to the leader. The leader is triggered if it receives at least  $(c/2)  \log  n$ trigger messages for that phase.
When the leader is triggered, it begins a new phase. Since there are $k/2$ triggering events, the expected number of trigger messages eventually generated is $(c \log n/r)(k/2) \geq c \log n$. Chernoff bounds imply that at least $(c/2) \log n$ trigger messages will be generated w.h.p. Alternatively, w.h.p., the number of trigger messages received by the leader will not exceed $(c/2) \log n$ until at least $k/8$ events have occurred, as this would imply twice the expected number.
We can conclude that w.h.p. the leader will trigger the next phase after $1/4$ of the outgoing edges to low-degree nodes have been found.
 
\begin{lemma}
\label{l:threshold}
When the leader receives  $(c/2)  \log n$ messages with the current phase number, w.h.p, at least 1/4 of the outgoing edges to low-degree nodes have been added to $Found_L$ lists.
\end{lemma}

\subsection{Proof of the main theorem:}
~~Here we prove Theorem \ref{theorem} as it applies to computing the spanning tree of a connected network with a pre-selected leader.

\begin{lemma} \label{l:AB}
W.h.p., after each phase except perhaps the first, either \textbf{(A)} A high-degree node (and star node) is added to $T$ or \textbf{(B)} $T$ is expanded so that the number of outgoing edges to low-degree nodes is reduced by a 1/4 factor (or the algorithm terminates with a spanning tree).
\end{lemma}

\begin{proof}
By Lemma \ref{aftersearch} there are three possible results from the {\bf Search} phase. If a sampled outgoing edge to a high-degree node is found, this edge will be added to the $Found_O$ list of its endpoint in $T$. If the {\bf Search} phase ends in fewer than $2c \log n$ edges found and none of them are to high degree nodes,  then w.h.p. these are all the outgoing edges to low-degree nodes, these edges will all be added to some $Found_L$.   If there are no outgoing edges, the algorithm terminates and a spanning tree has been found. If the third possible result occurs, then there are at least $2 \log n$ outgoing edges, half of which go to low-degree nodes. By Lemma \ref{l:threshold}, the leader will trigger the next phase and it will do so after at least $1/4$ of the outgoing edges to low-degree nodes have been added to $Found_L$ lists.  

By Observation \ref{o:expand}, all the endpoints of the edges on the $Found$ lists will be added to $T$ in the next phase, and there is at least one such edge or there are no outgoing edges and the spanning tree has been found. When $Expand$ is called in the next phase,  $T$ acquires a new high degree node in two possible ways, either because an outgoing edge on a $Found$ list is to a high-degree node or because the recursive $Expand$ on outgoing edges to low-degree edges eventually leads to an external high-degree node.  In either case, by the Invariant, $T$ will acquire a new star node as well as a high-degree node.  Also by the Invariant, all outgoing edges must come from high-degree nodes. Therefore, if no high-degree nodes are added to $T$ by Expand, then no new outgoing edges are added to $T$.  On the other hand, $1/4$ of the outgoing edges to low-degree nodes have become non-outgoing edges as their endpoints have been added to $T$. So we can conclude that the number of outgoing edges to low-degree nodes have been decreased by 1/4 factor.
\end{proof}

It is not hard to see:
\begin{lemma} \label{l:countPhases}
The number of phases is bounded by $O(\sqrt{n} \log^{1/2} n)$.
\end{lemma}
\begin{proof} By Lemma \ref{l:AB}, every phase except perhaps the first, is of type A or type B.
Chernoff bounds imply that w.h.p.,  the number of star nodes does not exceed its expected number  ($\sqrt{n} / \log^{1/2} n$) by more than a constant factor, hence there are no more than $O(\sqrt{n} / \log^{1/2} n)$ phases of type A. Before and after each such phase, the number of outgoing edges to low-degree nodes is reduced by at least a fraction of $1/4$; hence, there are no more than $\log_{4/3} n^2=O(\log n)$ phases of type B between phases of type A.
\end{proof}

Finally, we count the number of messages needed to compute the spanning tree. 
\begin{lemma} \label{l:countMessages}
The overall number of messages is $O(n^{3/2} \log^{3/2} n)$.
\end{lemma}

\begin{proof} 
The initialization requires $O(\sqrt{n} \log^{3/2} n)$ messages from $O(n)$ low-degree nodes and $O(n)$ messages from each of $O(\sqrt{n} / \log^{1/2} n)$ stars. In each phase, $Expand$ requires a number of messages which is linear in the size of $T$ or $O(n)$, except that newly added low-degree and star nodes send to their neighbors when they are added to $T$, but this adds just a constant factor to the initialization cost. $FindAny$ is repeated $O(\log n)$ times for a total cost of $O(n \log n)$ messages. $ApproxCut$ requires the same number. The Threshold Detector requires only $O(\log n)$ messages to be passed up $T$ or $O(n \log n)$ messages overall. Therefore, by Lemma \ref{l:countPhases} the number of messages over all phases is $O(n \log^{3/2} n)$.
\end{proof}

Theorem \ref{theorem} for spanning trees in connected networks with a pre-selected leader follows from Lemmas \ref{l:countMessages} and \ref{l:countPhases}.

\subsection{Proof of $ApproxCut$ Lemma} 
\label{s:approxcut}
\begin{proof}
Let $W$ be the set of the outgoing edges. For a fixed $z$ and $i$, we have:
\[ Pr(h_{z, i}(T) = 1) = Pr(\textit{an odd number of edges in W hash to } [2^i]) \geq \]
\[ Pr(\exists ! \textit{ e} \in W \textit{hashed to } [2^i]).\]
This probability is at least $1/16$ for $i = l - \lceil \log |W| \rceil - 2$ (Lemma 5 of \cite{king2015construction}). Therefore, since $X_j = \sum_{z = 1}^{c \log n}{h_{z, j}}$ (from pseudocode), $E[X_j] = \sum E[h_{z, j}] \geq c \log n / 16$, where $j =  l - \lceil \log |W| \rceil - 2$. Note that $ j = l - \lceil \log |W| \rceil - 2$ means that $\frac{2^l}{2^{j + 3}} < |W| < \frac{2^l}{2^{j + 1}}$. Consider $j - 4$. Since the probability of an edge being hashed to $[2^{j - 4}]$ is $\frac{2^{j - 4}}{2^l}$, we have
\[Pr(h_{z, j-4}(T) = 1) \leq Pr(\exists e \in W \textit{hashed to } [2^{j - 4}]) = |W| \frac{2^{j-4}}{2^l} \leq \frac{1}{2^{5}} \leq \frac{1}{32}.\]
Thus, $E[X_{j-4}] \leq c\log n/32$. Since an edge that is hashed to $[2^{j - k}]$ (for $k > 4$) is already hashed to $[2^{j - 4}]$, we have:
\[Pr(h_{z, j-4}(T) = 1 \vee \ldots \vee h_{z, 0}(T) = 1) \leq Pr(\exists e \in W \textit{hashed to } [2^{j - 4}] or \ldots or [2^{0}])) =\]
\[Pr(\exists e \in W \textit{hashed to } [2^{j - 4}]) = \frac{1}{32}.\]
Let $y_z$ be 1 if $h_{z, j-4}(T) = 1 \vee \ldots \vee h_{z, 0}(T) = 1$, and 0 otherwise. Also, let $Y = \sum_{z = 1}^{c \log n} y_z$. We have$E[Y] \leq c \log n / 32$. Also, for any positive integer $a$, 
\[ Pr(X_{j - 4} > a \vee \ldots \vee X_{0} > a) \leq Pr(Y > a).\]
From Chernoff bounds:
\[Pr(X_j < (3/4) c \log n /16) = Pr(X_j < (3/4) E[X_j]) < 1/n^{c'}\] and,
\[Pr(X_{j-4} > (3/2) c \log n / 16 \vee \ldots \vee X_{0} > (3/2) c \log n / 16) \leq Pr(Y > (3/2) c \log n / 16) = \]
\[Pr(Y > (3/2) c \log n / 32) < Pr(Y > (3/2) E[Y]) < 1/n^{c'}.\]
Therefore, by finding the smallest $i$ (called $min$ in pseudocode) for which $X_i > (3/2) c \log n/ 16$, w.h.p. $min$ is in $[j - 3, j]$. As a result, $2 |W| \leq 2^{l - min} \leq 64 |W|$. Therefore, $ |W|/ 32 \leq 2^{l - min} / 64 \leq |W|$. \par 
Furthermore, broadcasting each of the $O(\log n)$ hash functions and computing the corresponding vector takes $O(n)$ messages; so, the lemma follows.
\end{proof}
\old{\end{lemma}}

\subsection{Pseudocode}\label{s:pseudocode}
\begin{algorithm}
\caption{Initialization of the spanning tree algorithm}
\begin{algorithmic}[1]
\Procedure{Initialization}{}
\label{AinitAlg}
\State Every node selects itself to be a \emph{star} node with probability of $1 / \sqrt{n \log n}$.
\State Nodes that have degree $< \sqrt{n} \log^{3/2} n$ are \emph{low-degree} nodes. Otherwise, they are \emph{high-degree} nodes. (Note that they may also be star nodes at the same time.)
\State Star nodes send $\langle \textit{Star} \rangle$ messages to all of their neighbors.
\State Low-degree nodes send $\langle \textit{Low-degree} \rangle$ messages to all of their neighbors (even if they are star nodes too).
\EndProcedure
\end{algorithmic}
\end{algorithm}

\def\NoNumber#1{{\def\alglinenumber##1{}\State #1}\addtocounter{ALG@line}{-1}}

\begin{algorithm}
\caption{Asynchronous protocol for the leader to find a spanning tree.}
\begin{algorithmic}[1]
\Procedure{FindST-Leader}{}
\label{findMST}
\State Leader initially adds all of its incident edges to its $Found_L$ list. \texttt{// By exception leader does not need to differentiate between $Found_L$ and $Found_O$}
\State {$i \leftarrow 0$}
\Repeat {\textbf{ (Phase $i$)}}
\State $i \leftarrow i + 1$.
\State Leader calls $Expand()$. \texttt{// Expansion}
\label{expandLine}
\NoNumber{\texttt{// Search and Sampling:}}
\State $counter \leftarrow 0, A \leftarrow \emptyset$.
\label{firstlineafter}
\While{$counter <  48 c \log n$}
\label{sampleloop}
\State $FindAny(E \setminus A)$.
\If {$FindAny$ is successful and finds an edge $(u, v)$ ($u \in T$ and $v \notin T$)}
\State $u$ sends a message to $v$ to query $v$'s degree, and sends it to the leader.
\State $u$ adds $(u, v)$ to either $Found_L(u)$ or $Found_O(u)$ based on $v$'s degree.
\EndIf
\State $counter \leftarrow counter + 1$.
\EndWhile
\If {$|A| = 0$} 
\State {\bf terminate the algorithm} as there are no outgoing edges.
\label{nooutgoing}
\ElsIf {$|A| < 2 \log n $ (few edges) \textbf{or} $\exists (u, v) \in A$ s.t. $v$ is high-degree}
\label{onwards}
\State Leader starts a new phase to restore the Invariant.
\Else {\textit{ }(at least half of the outgoing edges are to low-degree nodes) \texttt{// Wait:} }
\State $r \leftarrow ApproxCut() / 2$.
\label{approxline}
\State Leader calls $ThresholdDetection(r)$.
\State Leader waits to trigger and then starts a new phase.
\label{waitperiod}
\EndIf
\Until
\EndProcedure
\end{algorithmic}
\end{algorithm}

\begin{algorithm}
\caption{Given $r$ at phase $i$, this procedure detects when nodes in $T$ receive at least $r/4$ $\langle Low-degree \rangle$ messages over outgoing edges. $c$ is a constant.}
\begin{algorithmic}[1]
\Procedure{ThresholdDetection}{}
\State Leader calls Broadcast($\langle \textit{Send-trigger}, r, i \rangle$). 
\State When a node $u \in T$ receives $\langle \textit{Send-trigger}, r, i \rangle$, it first participates in the broadcast. Then, for every event, i.e. every edge $(u, v) \in {Found(u)}_{L}$ such that $v \notin \textit{T-neighbor}(u)$, $u$ sends to its parent a $\langle Trigger, i \rangle $ message with probability of $c \log n / r$.
\State A node that receives $\langle Trigger, i \rangle $ from a child keeps sending up the message until it reaches the leader. If a node receives an $\langle \textit{Expand} \rangle$ before it sends up a $\langle \textit{Trigger}, i \rangle$, it discards the $\langle \textit{Trigger}, i \rangle$ messages as an Expand has already been triggered.
\State Once the leader receives at least $c \log n/2$ $\langle \textit{Trigger}, i \rangle$ messages, the procedure \textbf{terminates} and the control is returned to the calling procedure.
\EndProcedure
\end{algorithmic}
\end{algorithm}

\begin{algorithm} \label{expand}
\caption{Leader initiates Expand by sending $\langle \textit{Expand} \rangle$ to all of its children. If this is the first Expand, leader sends to all of its neighbors. Here, $x$ is any non-leader node.}
\begin{algorithmic}[1]
\Procedure{Expand}{}
\label{expandAlg}
\State When node $x$ receives an $\langle Expand \rangle $ message over an edge $(x, y)$:
\State $x$ adds $y$ to $\textit{T-neighbor}(x)$.
\If{$x$ is not in $T$}
\State The first node that $x$ receives $\langle \textit{Expand} \rangle$ from becomes $x$'s parent. \texttt{//$x$ joins $T$}
\If{$x$ is a high-degree node \textbf{and} $x$ is not a star node}
\State It waits to receive a $\langle \textit{Star} \rangle$ over some edge $e$, then adds $e$ to $Found_O(x)$.
\label{expandWait}
\State It forwards $\langle \textit{Expand} \rangle$ over edges in $Found_L(x)$ and $Found_O(x)$ (only once in 
\Statex \qquad \qquad \quad case an edge is in both lists), then removes those edges from the Found lists.
\Else \textit{ }($x$ is a low-degree or star node)
\State It forwards the $\langle \textit{Expand} \rangle$ message to all of its neighbors.
\EndIf
\Else \textit{ }($x$ is already in $T$)
\State If the sender is not its parent, it sends back $\langle \textit{Done-by-reject} \rangle$.  Else, it forwards 
\Statex \qquad \qquad $\langle \textit{Expand} \rangle$ to its children in $T$, over the edges in $Found_L(x)$ and $Found_O(x)$,
\Statex \qquad \qquad then removes those edges from the Found lists.
\EndIf
\Statex \parbox[t]{\dimexpr\linewidth-\algorithmicindent}{ \texttt{// Note that if $x$ added more edges to its Found list after forward of $\langle \textit{Expand} \rangle$, the new edges will be dealt with in the next Expand.} \strut}
\State \parbox[t]{\dimexpr\linewidth-\algorithmicindent}{When a node receives $\langle Done \rangle $ messages (either $\langle \textit{Done-by-accept} \rangle$ or $\langle \textit{Done-by-reject} \rangle$) from all of the nodes it has sent to, it considers all nodes that have sent $\langle \textit{Done-by-accept} \rangle$ as its children. Then, it sends up $\langle \textit{Done-by-accept} \rangle$ to its parent. \strut}
\State The algorithm terminates when the leader receives $\langle Done \rangle $ from all of its children.
\EndProcedure
\end{algorithmic}
\end{algorithm}

\begin{algorithm}
\caption{Approximates the number of outgoing edges within a constant factor. $c$ is a constant.}
\begin{algorithmic}[1]
\Procedure{ApproxCut}{$T$}
\State Leader broadcasts $c \log n$ random 2-wise independent hash functions defined from $[1, n^{2c}] \rightarrow [2^l]$.
\State For node $y$, and hash function $h_z$ vector $\overrightarrow{h_z}(y)$ is computed where $h_{z,i}(y)$ is the parity of incident edges that hash to $[2^i]$, $i = 0, \ldots, l$.
\State For hash function $h_z$, $\overrightarrow{h_z}(T) = \oplus_{y \in T}{\overrightarrow{h_z}(y)}$ is computed in the leader.
\State For each $i = 0, \ldots, l$, $X_i = \sum_{z = 1}^{c \log n}{h_{z, i}(T)}$.
\State Let $min$ be the smallest $i$ s.t. $X_i \geq (3/4) c \log n / 16$.
\State Return $2^{l - min} / 64$.
\EndProcedure
\end{algorithmic}
\end{algorithm}

\section{Finding MST with $o(m)$ asynchronous communication}
\label{MSTconstruction}
The MST algorithm implements a version of the GHS algorithm which grows a forest of disjoint subtrees (``fragments'')  of the MST in parallel. We reduce the message complexity of GHS by using $FindMin$ to find minimum weight outgoing edges \emph{without} having to send messages across every edge. But, by doing this, we require the use of a spanning tree to help synchronize the growth of the fragments.

 Note that GHS nodes send messages along their incident edges for two main purposes: (1) to see whether the edge is outgoing, and (2) to make sure that fragments with higher rank are slowed down and do not impose a lot of time and message complexity. Therefore, if we use $FindMin$ instead of having nodes to send messages to their neighbors, we cannot make sure that higher ranked fragments are slowed down. Our protocol works in phases where in each phase only fragments with smallest ranks continue to grow while other fragments wait. A spanning tree is used to control the fragments based on their rank. (See \cite{gallager1983distributed} for the original GHS.) \\ \\
{\it Implementation of FindMST:} Initially, each node forms a fragment containing only that node which is also the leader of the fragment and fragments all have rank zero. A \emph{fragment identity} is the node ID of the fragment's leader; all nodes in a fragment know its identity and its current rank. Let the pre-computed spanning tree $T$ be rooted at a node $r$, All fragment leaders wait for instructions that are broadcast by $r$ over $T$. \par 
The algorithm runs in phases. At the start of each phase, $r$ broadcasts the message $\langle \textit{Rank-request} \rangle$ to learn the current minimum rank among all fragments after this broadcast. Leaves of $T$ send up their fragment rank. Once an internal node in $T$ receives the rank from all of its children (in $T$) the node sends up the minimum fragment rank it has received including its own. This kind of computation is also referred to as a \emph{convergecast}. \par
Then, $r$ broadcasts the message $\langle Proceed, minRank \rangle$ where $minRank$ is the current minimum rank among all fragments. Any fragment leader that has rank equal to $minRank$, \emph{proceeds} to finding minimum weight outgoing edges by calling FindMin on its own fragment tree. These fragments then send a $\langle Connect \rangle$ message over their minimum weight outgoing edges. When a node $v$ in fragment $F$ (at rank $R$) sends a $\langle Connect \rangle$ message over an edge $e$ to a node $v'$ in fragment $F'$ (at rank $R'$), since $R$ is the current minimum rank, two cases may happen: (Ranks and identities are updated here.)
\begin{enumerate}
\item \textbf{$R < R'$}: In this case, $v'$ answers immediately  to $v$ by sending back an $\langle Accept \rangle$ message, indicating that $F$ can merge with $F'$. Then, $v$ initiates the merge by changing its fragment identity to the identity of $F'$, making $v'$ its parent, and broadcasting $F'$'s identity over fragment $F$ so that all nodes in $F$ update their fragment identity as well. Also, the new fragment (containing $F$ and $F'$) has rank $R'$.
\item \textbf{$R = R'$}: $v'$ responds  $\langle Accept \rangle$ immediately to $v$ if the minimum outgoing edge of $F'$ is $e$, as well. In this case, $F$ merges with $F'$ as mentioned in rule 1, and the new fragment will have $F'$'s identity. Also, both fragments increase their rank to $R' + 1$. \\
Otherwise, $v'$ does not respond to the message until $F'$'s rank increases. Once $F'$ increased its rank, it responds via an $\langle Accept \rangle$ message, fragments merge, and the new fragment will update its rank to $R'$.
\end{enumerate}
The key point here is that fragments at minimum rank are not kept waiting. Also, the intuition behind rule 2 is as follows. Imagine we have fragments $F_1, F_2, ..., F_k$ which all have the same rank and $F_i$'s minimum outgoing edge goes to $F_{i+1}$ for $i \leq k - 1$. Now, it is either the case that $F_k$'s minimum outgoing edge goes to a fragment with higher rank or it goes to $F_k$. In either case, rule 2 allows the fragments $F_{k - 1}, F_{k-2}, \ldots$ to update their identities in a cascading manner right after $F_k$ increased its rank. \par 
When all fragments finish their merge at this phase they have increased their rank by at least one. Now, it is time for $r$ to star a new phase. However, since communication is asynchronous we need a way to tell whether all fragments have finished. In order to do this, $\langle Done \rangle $ messages are convergecast in $T$. Nodes that were at minimum rank send up to their parent in $T$ a $\langle Done \rangle $ message only after they increased their rank and received $\langle Done \rangle $ messages from all of their children in $T$. \par 
As proved in Lemma \ref{maintheo}, this algorithm uses $\tilde{O}(n)$ messages.

\begin{algorithm}
\caption{MST construction with $\tilde{O}(n)$ messages. $T$ is a spanning tree rooted at $r$.}
\begin{algorithmic}[1]
\Procedure{FindMST}{}
\State All nodes are initialized as fragments at rank 0.
\NoNumber{\texttt{// Start of a phase}}
\State $r$ calls Broadcast($\langle \textit{Rank-request} \rangle$), and $minRank$ is computed via a convergecast.
\label{levelreq}
\State $r$ calls Broadcast($\langle Proceed, minRank \rangle $).
\State Fragment leaders at rank $minRank$ that have received the $\langle Proceed, minRank \rangle$ message, call FindMin. Then, these fragments merge by sending \emph{Connect} messages over their minimum outgoing edges. If there is no outgoing edge the fragment leader \textbf{terminates the algorithm}.
\State
Upon receipt of $\langle Proceed, minRank \rangle$, a node $v$ does the following: \\
If it is a leaf in $T$ at rank $minRank$, sends up $\langle Done \rangle$ after increasing its rank. \\
If it is a leaf in $T$ with a rank higher than $minRank$, it immediately sends up $\langle Done \rangle$. \\
If it is not a leaf in $T$, waits for $\langle Done \rangle $ from its children in $T$. Then, sends up the $\langle Done \rangle $ message after increasing its rank.
\State $r$ waits to receive $\langle Done \rangle $ from all of its children, and starts a new phase at step \ref{levelreq}.
\EndProcedure
\end{algorithmic}

\end{algorithm}

\begin{lemma}
\label{maintheo}
FindMST uses $O(n \log^3 n / \log \log n)$ messages and finds the MST w.h.p.
\begin{proof}
All fragments start at rank zero. Before a phase begins, two broadcasts and convergecasts are performed to only allow fragments at minimum rank to proceed. This requires $O(n)$ messages. In each phase, finding the minimum weight outgoing edges using FindMin takes $O(n \log^2 n / \log \log n)$ over all fragments. Also, it takes $O(n)$ for the fragments to update their identity since they just have to send the identity of the higher ranked fragment over their own fragment. As a result, each phase takes $O(n \log^2 n / \log \log n)$ messages. \par 
A fragment at rank $R$ must contain at least two fragments with rank $R - 1$; therefore, a fragment with rank $R$ must have at least $2^R$ nodes. So, the rank of a fragment never exceeds $\log n$. Also, each phase increases the minimum rank by at least one. Hence, there are at most $\log n$ phases. As a result, message complexity is $O(n \log^3 n / \log \log n)$.

\end{proof}
\end{lemma}

From Lemma \ref{maintheo}, Theorem \ref{theorem} for minimum spanning trees follows.

\section{Minimum Spanning Forest}
\label{disconCase}
In this section we provide an algorithm to find a minimum spanning forest when the input graph is disconnected. The algorithm still computes an MST if the graph is connected. Unlike the old algorithm in Section \ref{STconstruction}, the algorithm presented here does not require a pre-selected leader.\par
In order to simplify the description of the algorithm, we use the alternative way for counting events as mentioned in Section \ref{defs}. In fact, we assume that nodes send back acknowledgment messages when they receive a $\langle \textit{Low-degree} \rangle$ message. Therefore, we will not need the $\textit{T-neighbor}$ data structure.\par

The main idea for our algorithm is to deal with components based on their sizes. A component is considered \textit{big} if it has at least $ \sqrt{n} \log^{3/2} n$ nodes, and \emph{small} otherwise. With high probability, there is a star node inside all big components. However, a small component does not necessarily contain a star node. Moreover, small components cannot have a high-degree node since a node is high-degree if it has at least $\sqrt{n} \log^{3/2} n$ neighbors. \par
Nodes, however, do not know the size of the component they belong to. To overcome this obstacle, we are going to have two types of nodes that run different algorithms as follows.
\begin{itemize}
\item Star nodes run the FindST-leader protocol.
\item Low-degree nodes that have not been selected as a star-node (called low-degree non-star nodes) run the GHS algorithm.
\end{itemize}
Note that all nodes still send the initial $\langle \textit{Low-degree} \rangle$ and $\langle Star \rangle $ messages.\newline
Having nodes that run different algorithms in the same component will eventually result in a situation where a node is asked to participate in more that one protocol (initiated by different nodes). We call such an event a \textit{collision}. We will later explain how we can handle these collisions so that a node can decide which protocol it should be a part of. \par

The intuition behind running different protocols is as follows. If a component is small, all of its nodes are low-degree. Moreover, as mentioned before, a small component may not a have star node. Therefore, a small component should be able to compute the MST on its own. Since there is no high-degree node in a small component, the total number of edges in the component is $O(n^{3/2} \log^{3/2} n)$. Therefore, all nodes can run the GHS algorithm which uses messages proportional to the number of edges in the component, and find the MST directly. However, to ensure that the message complexity still remains $o(m)$ we will not allow any high-degree node to participate in a GHS protocol. Therefore, the GHS protocol will complete successfully if and only if it is running in a component with no high-degree and no star node.\par
On the other hand, if a component is big it has a star node w.h.p.. Therefore, in each big component, we can have a star node that acts as the leader of that component. Following the Find-ST protocol, the nodes of a big component will first find a spanning tree over the component's subgraph, and then compute the minimum spanning tree. \par
In order to implement this idea, we need a set of rules to handle the collision between protocols. Here, we do not consider an exchange of message between two nodes that are running the GHS protocol to be a collision. The GHS protocol can deal with these cases itself. \par 
In this section, we assume that fragments which grow using the FindST-leader, each have a \emph{fragment ID}. This ID is the node ID of the star node that initiated the protocol. In order to resolve collisions, there are three cases that we should consider. Imagine that node $a$ sends a message regarding the protocol it is running to node $b$. The collision resolution rules are as follows: \
\begin{enumerate}
\item When $a$ is running the FindST-leader protocol and $b$ is running the GHS protocol: in this case $b$ will stop participating in the GHS protocol and will become a part the FindST-leader protocol that $a$ is running.

\item When $a$ is running GHS and $b$ is running FindST-leader: $b$ will never respond to the message; therefore, $a$'s GHS protocol will never terminate.

\item When $a$ and $b$ are both running the FindST-leader protocols with different fragment IDs : in this case if $a$'s fragment ID is higher than $b$'s fragment ID, it can proceed and take over node $b$. Otherwise, $b$ tells $a$ that $a$'s fragment should stop running the protocol. (We will elaborate on this later.)
\end{enumerate}

Having this set of rules, we can always preserve the following invariant:\\
{\bf Invariant:} In a component with at least one star node, the star node with maximum ID (among all star nodes in that component) will become responsible for computing the MST. 

\old{The important thing here is how to handle the situations when two algorithms collide. In fact, we  GHS and FindST-Leader algorithm, and also between two FindST-Leader algorithms executed by different star nodes. \par 
If a component does not have a star or high-degree node, the minimum spanning tree in that component is found directly as all nodes in the component are running the GHS protocol. Now, imagine that a component has more than one star node and also low-degree non-star nodes. In this case, the policy is to always give priority to the FindST-Leader protocols which are run by star nodes. The reason for this is that first we want to make sure that no high-degree node runs the GHS protocol, and second that there is always a star node responsible for finding ST and MST in the component. Also, if two FindST-Leader protocols collide priority is given to the one whose leader has a smaller ID.}

\noindent
\textit{Description of the algorithm:} First, all nodes follow the initialization protocol (Algorithm \ref{AinitAlg}) as in the old algorithm. Then, each star node runs the FindST-Leader protocol to expand its fragment tree via \textit{expansion}, \textit{search and sampling}, and \textit{waiting} where all of the messages exchanged during these three phases are labeled by the fragment ID, i.e., the ID of the star node that is the leader of the fragment. \par
All low-degree non-star nodes run the GHS protocol. Note that our assumption is that all nodes are awakened simultaneously. \par 
If a component is composed of only low-degree non-star nodes, these nodes, which are all running the GHS protocol, will find the minimum spanning tree in that component directly. However, in order to obtain sublinear message complexity, we ensure that no high-degree node participates in a GHS protocol. To this end, when a high-degree node receives a message from a node running the GHS protocol, it will never respond to the message; hence, delaying the protocol forever. Note that existence of a high-degree node guarantees the existence of a star node w.h.p.. Therefore, a star node will eventually find the MST in this component. \par
As mentioned before, in order to resolve collisions we require each fragment to have an ID. Initially, each star node belongs to a singleton fragment with an ID equal to that of the star node itself.  Although initially non-star nodes do not belong to any fragment, we assume that they belong to a fragment with an ID of 0. This assignment allows fragments that are lead by a star node to take over the nodes that have not joined a fragment yet, using the third collision resolution rule (prioritizing higher fragment IDs). \par 
It is worth mentioning that a node may be contacted by a node in a different fragment through initialization messages, messages for querying the degree, or expansion messages. We do not label initialization and degree-querying messages. Therefore, a node that receives these types of messages will not consider it a collision.  This will not affect the analysis of the algorithm.\par 
In order to enforce the collision resolution rules and to keep the message complexity sublinear, we use a new method for the expansion phase. Expand-MultiLeader (Algorithm \ref{expandMulti}) allows the fragments IDs to be updated without using a lot of messages. Expand-MultiLeader will replace the old Expand algorithm in Section \ref{STconstruction}.\\ 

\noindent
{\it Implementation of Expand-MultiLeader:}
The expansion messages used in Expand-MultiLeader are accompanied by the fragment identities ($\langle Expand, ID \rangle$). Every node $v$ keeps a variable $vID$ which is the ID of the fragment it belongs to at the moment. We also use $vID$ to refer to the whole fragment whose ID is equal to $vID$. As mentioned before, for all non-star nodes $vID$ is initialized to 0. We define the \emph{reject list} data structure as follows:\par

$Reject(v)$: Contains the list of edges over which $v$ rejected an expansion message from another fragment for the first time. An edge is only added to this list if the reason of rejection is lower ID value of the other fragment. We emphasize that the edge is added only on the first occurrence of this event for each fragment.

When node $x$ (currently in fragment $xID$) receives an expansion message of $\langle Expand, tID \rangle$ from some node $t$ in fragment $tID$, the algorithm handles this message by comparing $xID$ and $tID$. There are three cases: \\
\begin{enumerate}
\item $tID < xID$: Fragment $xID$ refuses to join fragment $tID$, and signals fragment $tID$ to stop running the FindST-Leader protocol. To this end, $x$ responds to $t$'s expansion message by sending back $\langle \textit{Rejected-lower-ID}, tID \rangle$ (line \ref{rejHigher}). The leader of fragment $tID$ will eventually be notified that its fragment has contacted a fragment with higher ID when the $\langle \textit{Rejected-lower-ID}, tID \rangle$ message reaches the leader via a convergecast (line \ref{rejConvergecast}). This causes fragment $tID$ to stop running the FindST-Leader protocol (line \ref{stopProto}). \par 
Moreover, fragment $xID$ should take over the fragment $tID$. However, since fragment $xID$ may not be running an expansion at the moment, node $x$ has to \emph{remember} to take over fragment $tID$ in the next expansion. For this, $x$ adds the edge $(x,t)$ to $Reject(x)$ only if this is the \textit{first time} that $x$ sees an expansion message with identity $tID$. The reason for the first time condition is that since all nodes of fragment $tID$ are already connected via a tree structure, having access to at least one of those tree nodes would suffice to allow fragment $xID$ to take over the whole fragment later. During the next expansion, $x$ will forward the expansion messages over the edges in $Reject(x)$, as well.\par 
Notice that in this case $x$ responds to an expansion message immediately even if it is participating in another expansion at the moment. This prevents fragments from waiting in a loop for each others' expansion to finish. 

\old{
An example of this is shown in Figure \ref{LoopExample}. Consider three fragments $A$, $B$, and $C$ such that $A$'s ID is larger than $B$'s ID, and $B$'s ID is larger than $C$'s ID. In this situation, C does not respond to A since it has to finish its expansion first; however, if $\langle Expand, C \rangle$ and $\langle Expand, B \rangle$ are not rejected immediately, a loop is created.

\begin{figure*}
\caption{Fragments waiting in a loop}
\label{LoopExample}
\centering
\includegraphics[width=.75\textwidth]{LoopExample.eps}
\end{figure*}
}

\begin{figure*}
\caption{Example of quadratic message complexity}
\label{ManyMsgs}
\centering
\includegraphics[width=.75 \textwidth]{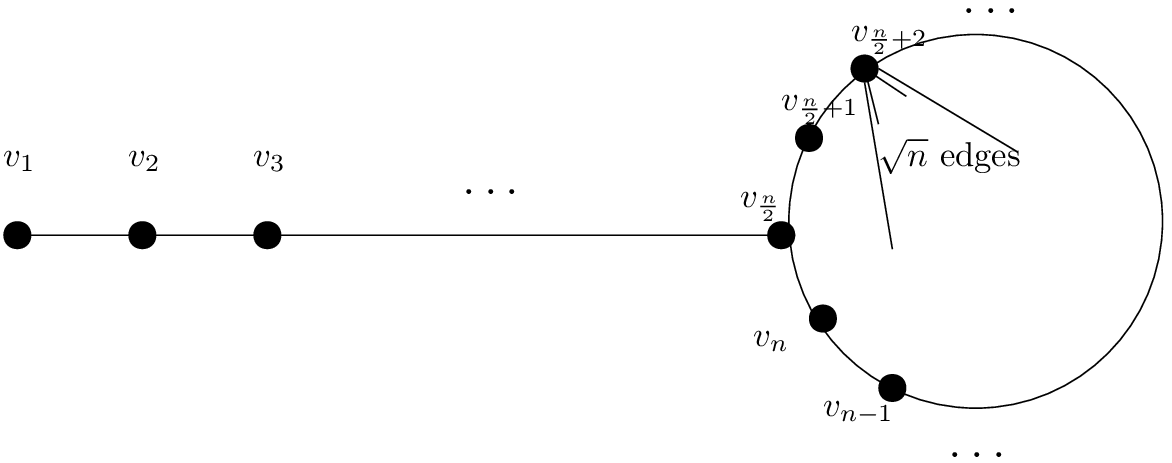}
\end{figure*}

\item $tID > xID$: In this case, $x$ joins fragment $tID$ by updating its fragment identity from $xID $ to $tID$ and setting the first node from which it has received $\langle Expand, tID \rangle$ as its parent (line \ref{joinT}). \par 
Notice that unlike the previous case, $x$, if participating in an expansion, first waits to hear the response to all expansion messages it has sent, then updates its children, and only after that it handles the $\langle Expand, tID \rangle$ message. This will ensure that the expansion messages are always forwarded over a \textit{tree} structure when it is updating the ID of another fragment. Figure \ref{ManyMsgs} shows that without this wait, the algorithm needs $\Omega(n^2)$ messages to finish. Imagine that in Figure \ref{ManyMsgs} the nodes have IDs $v_1 > v_2 >\ldots > v_n$. Let $C$ be the set of nodes $\{v_{n/2}, \ldots, v_n\}$ and all nodes in $C$ are low-degree and have $\sqrt{n}$ edges to other nodes in $C$. There have to be $\Theta(\sqrt{n \log n})$ star nodes on the path from $v_1$ to $v_{n/2}$, w.h.p.. Let us name these star nodes $s_1$ to $s_k$ from left to right, where $k = \Theta(\sqrt{n \log n})$. Assume $s_k$'s expansion message reaches to $v_{n/2}$ and is forwarded to the nodes in $C$, but before any expansion message is answered, $s_{k - 1}$'s expansion message reaches to $v_{n/2}$. Therefore, if $v_{n/2}$ does not wait for the last expansion to finish, it has to forward $\langle Expand, s_{k-1}\rangle$ to nodes in $C$ again. Besides, these are all low-degree nodes and will forward the message over all of their incident edges which requires $\Omega(n^{3/2})$ messages. Repeating this for $s_{k - 2}$ to $s_{1}$ will result in $\Omega(n^2)$ messages.

Now, there are two sub-cases for Case 2, based on whether or not $x$ is part of some fragment upon receiving $\langle Expand, tID \rangle$:\\
\begin{itemize}
\item[2.a] If $x$ is part of some other fragment $X$ (line \ref{notFirstTime}), fragment $X$ should also become a part of fragment $tID$. So, $x$ forwards  $\langle Expand, tID \rangle$ to all of its neighbors in $X$ except the one that it has received the expansion message from. This will result in all neighbors of $x$ to do the same recursively and all nodes in fragment $X$ will update their fragment identity to $tID$ and become part of fragment $tID$. Moreover, $x$ forwards the expansion message over edges in $Found(x)$ and $Reject(x)$.

\item[2.b] If $x$ did not belong to any fragment before (line \ref{isFirstTime}), the expansion message is handled just like the old Expand algorithm. Note that in this case $Reject(x)$ is an empty list. If the node is high-degree and is not a star node (line \ref{highNonStar}), it waits to receive a star message, updates its $Found_O$ list, and then forwards the message to the edges in its Found lists. Else, if $x$ is a low-degree or star node, it forwards the message to all of its neighbors in $G$.
\end{itemize} 
\item $tID = xID$: In this case, $x$ and $t$ are in the same fragment. If $t$ is $x$'s parent in fragment $tID$, $x$ forwards the expansion message to its children in the fragment, over the edges in Found lists and also over the edges in $Reject(x)$. Otherwise, the sender of the message is not $x$'s parent; therefore, $x$ responds back $\langle \textit{Reject-same-tree}, tID \rangle$.
\end{enumerate}

Once the fragment identities are compared and the expansion message is handled, $x$ waits to hear back from all the nodes it has forwarded to. Then, $x$ removes from $Reject(x)$ and the Found lists any edges it has forwarded over. Afterwards, $x$ updates its children to be all of the nodes that it has received $\langle Accept, tID \rangle$ from (line \ref{updateBranches}) in this expansion. A node only sends up an $\langle Accept, tID \rangle$ message in response to $\langle Expand, tID \rangle$ if either it is a leaf in the fragment tree or it has not received any $\langle \textit{Rejected-lower-ID}, tID \rangle$ from the nodes it has sent to in this expansion. If a node receives a $\langle \textit{Rejected-lower-ID}, tID \rangle$, it sends up $\langle \textit{Rejected-lower-ID},$ $tID \rangle$ to its parent. As a result, if any node in the fragment is rejected because of its low ID, the leader will be notified and will stop running the FindST-Leader protocol (line \ref{stopProto}).
\subsection{Correctness}
We argue about the correctness of the algorithm for the following two cases:
\begin{itemize}
\item If a component has at least one star node: In this case, because of the way we defined the collision resolution rules, eventually, the star node with the highest ID among other possible star nodes in that component will take over all other fragments and proceeds by running the FindST-leader protocol. Therefore, the correctness of the algorithm follows by the correctness of the old algorithm.

\item If a component does not have a star node: In this case, w.h.p. the component is small. Therefore, all nodes are low-degree and will run the GHS protocol. So, the correctness of the algorithm follows from the correctness of the GHS algorithm.
\end{itemize}
\subsection{Analysis}
As before, the overall number of initialization messages is still bounded by $O(n^{3/2} \log^{3/2} n)$. Also, GHS spends asymptotically the same number of messages as the number of edges in the graph running the protocol. Since no high-degree node participates in a GHS, the overall number of messages used for the GHS protocol is bounded by $O(n^{3/2} \log^{3/2} n)$. Now, we show that the overall number of messages used for running FindST-Leader's is also bounded by $O(n^{3/2} \log^{3/2} n)$. \par 
We define a \emph{successful expansion} to be one that completes without any node receiving a $\langle \textit{Rejected-lower-ID} \rangle $ message, which does not cause the leader to stop the FindST-Leader protocol. 
\begin{lemma}
\label{successfulExpands}
The number of successful expansions over all fragments is $O(\sqrt{n} \log^{1/2} n)$.
\begin{proof}
Any fragment can have at most $O(\log n)$ successful expansions before it causes the number of fragments to reduce, i.e., it causes two fragments with different identities to merge. The reason is that as we know from Lemma \ref{l:countPhases}, $O(\log n)$ expansions suffice for a fragment to find a new star (and a new fragment) with high probability, which will cause at least one fragment to stop running the protocol. Since the initial number of fragments is $O(\frac{\sqrt{n}}{\log^{1/2} n})$, the overall number of successful expansions is bounded by $O(\sqrt{n} \log^{1/2} n)$.
\end{proof}
\end{lemma}

In any fragment, each successful expansion is followed by one \textit{search and sampling} and at most one \textit{wait} where each of which these parts uses $O(n \log n)$ messages (proof in Lemma \ref{l:countMessages}). Therefore, from Lemma \ref{successfulExpands} , the overall number of messages used in the algorithm apart from expansion is bounded by $O(n^{3/2} \log^{3/2} n)$. \par 

Finally, we prove in the following claims that the overall number of messages used for expansion is bounded by the same amount.

\begin{claim}[1]
\label{foundListClaim}
The number of forwards over edges in the Found lists, over all nodes, is bounded by $O(n^{3/2} \log^{3/2} n)$. 
\begin{proof}
After an $\langle Expand, ID \rangle $ is forwarded over an edge in the Found list, regardless of the result, that edge is removed from the Found list (line \ref{emptyFoundlist}). Therefore, we only need to show that the overall number of edges added to Found lists of all nodes is bounded by $O(n^{3/2} \log^{3/2} n)$. An edge is added to the Found list of a node either because of a $\langle \textit{Low-degree} \rangle$ message, a $\langle Star \rangle$ message, or an edge found by FindAny in FindST-Leader. The first two types of messages are bounded by $O(n^{3/2} \log^{3/2} n)$. According to Lemma \ref{successfulExpands}, there are $\tilde{O}(\sqrt{n})$ successful expansions over all fragments and each of them is followed by at most $O(\log n)$ calls to FindAny. Therefore, FindAny is responsible for at most $\tilde{O}(\sqrt{n})$ edges in the Found lists and the claim follows.
\end{proof}
\end{claim}

\begin{claim}[2]
The number of forwards over edges in the Reject lists, over all nodes, is bounded by $O(n^{3/2} /  \log^{1/2} n)$.
\begin{proof}
A node $x$ that receives $\langle Expand, tID \rangle $ over edge $e$, and rejects it because $tID < xID$, may only add $e$ to its Reject list if this is the first time it receives an expansion message with identity $tID$. Therefore, a node $x$ adds $O(\sqrt{n} / \log^{1/2} n)$ (the initial number of fragments) edges to its Reject list over the course of the algorithm. Also, upon forwarding over some edge $e$ in the Reject list, the node removes $e$ from the list; therefore, the overall number of such forwards is $O(n^{3/2} / \log^{1/2} n)$
\end{proof}
\end{claim}

\begin{claim}[3]
The overall number of expansion messages forwarded, when a node receives an expansion message for the \emph{first time}, is bounded by $O(n^{3/2} \log^{3/2} n)$ over all nodes.
\begin{proof}
The first time that a node receives an $\langle \textit{Expand, ID} \rangle$ message, if it is high-degree, it forwards only over its Found list, and if it is a low-degree node, it forwards to up to $\sqrt{n} \log^{3/2} n$ nodes. Using the Claim (1), the claim follows.
\end{proof}
\end{claim}

Finally, the following claim will bound the number of forwards \emph{after the first time} and allows us to bound the message complexity of the algorithm.

\begin{claim}[4]
The number of forwards over branches of the old fragment in case that $tID > xID$ (line \ref{tochildrenless}), and also the number of forwards over branches of the current fragment in case that $tID = xID$ (line \ref{tochildrenequal}) is bounded by $O(n^{3/2} \log^{1/2} n)$.
\begin{proof}
\textit{Forwards over incident edges in the old fragment when $tID > xID$:} Over the whole algorithm, $O(\sqrt{n}/ \log^{1/2} n)$ leaders may have grown their fragments to a tree of size  $O(n)$. Assume that each node had a set of incident edges (including the one to its parent) in each of the fragments it belonged to over the course of the algorithm. Let $C$ be the collection of all of the incident edges of all nodes in these fragments. Note that in $C$, the same edge is repeated twice for each of its endpoints and could also be repeated up to $O(\sqrt{n}/ \log^{1/2} n)$ times as part of different fragments. Since there are $O(\sqrt{n}/ \log^{1/2} n)$ fragment trees and each of them has size at most $n$, size of $C$ is the sum of the degrees of all nodes in these trees which is bounded by $O(n^{3/2} / \log^{1/2} n)$.  \par 
Consider node $x$ that is part of a fragment $X$ with identity $xID$. When $x$ receives $\langle Expand, tID \rangle $ it forwards over to its neighbors in $X$ except the one it has received the expansion message from (line \ref{tochildrenless}). Whether the result of the forward is accept or reject, this is the last time that $x$ forwards over a part of the fragment $X$, \textit{as a node in} $X$. The reason is that right after this $x$ joins fragment $T$ and updates its children (line \ref{updateBranches}). Now, any future expansion message that $x$ receives, will be forwarded as a part of fragment $T$ (and not $X$). Therefore, the overall number of expansion messages that nodes forward over their incident edges, immediately after updating their identity, could not exceed the size of $C$ which is $O(n^{3/2} / \log^{1/2} n)$.\\ \\
\noindent
\textit{Forwards to children when $tID = xID$:} In this case, $x$ only forwards the message if the sender is its parent in $T$. In fact, if $x$ receives and forwards $\langle Expand, tID \rangle $ messages $k$ times, it has to be the case that $T$ has performed at least $k - 1$ \emph{successful} expansions. Otherwise, $T$ would have stopped before starting the $k^{th}$ expansion. Since the overall number of successful expansions is $O(\sqrt{n} \log^{1/2} n)$, the overall number of forwards over incident edges when $tID = xID$ is $O(n^{3/2} \log^{1/2} n)$.
\end{proof}
\end{claim}

Putting together Claims 1 to 4, we obtain the following lemma.

\begin{lemma}
\label{multiLeaderProof}
Expand-MultiLeader will result in no more than $O(n^{3/2} \log^{3/2} n)$ forwards of $\langle \textit{Expand, ID} \rangle$ messages over all fragments, before a spanning forest is constructed.
\end{lemma}

Therefore, we have the following theorem on building a minimum spanning forest in general input graphs.

\begin{theorem}
A minimum spanning forest can be constructed using  $O(n^{3/2} \log^{3/2} n)$ messages with asynchronous communication.
\end{theorem}

\noindent
\textit{\textbf{A note on analysis:}} Here, we provide an example that shows even when all star nodes are running the protocol in parallel, the time and the message complexity could still be as high as $\Theta(n^{3/2} \log^{3/2} n)$. Since we do not assume to know initially whether the graph $G$ is a connected graph or not, our example here is a connected graph. Consider the graph in Figure \ref{lollipop} where a complete graph $K_{n/2}$ is connected to two path graphs $P_{n/4}$ on the left and on the right, named $P_L$ and $P_R$ respectively. Assume that node IDs are $v_1 > v_2 > \ldots > v_n$. There have to be $\Theta(\sqrt{n \log n})$ star nodes in $P_L$ and in $K_{n/2}$ w.h.p.. As a result, $\Theta(n^{3/2} \log^{3/2} n)$ messages are guaranteed since the star nodes in $K_{n/2}$ send initialization messages to all of their neighbors. \par 
 We show that time complexity is $\Theta(n^{3/2} \log^{3/2} n)$, as well. Let the star nodes in $P_L$ be $s_1, s_2, \ldots, s_k$ from left to right, where $k = \Theta(\sqrt{n \log n})$. Let $\delta=1$ time step be the max delay. Suppose $s_k$'s expansion messages go all the way to the right in one time step, and span $K_{n/2}$ and $P_R$. Now, $s_{k - 1}$'s expansion reaches to $s_k$ and updates the identity of all of the nodes on the right in $O(n)$ time steps. Meanwhile, $s_{k-2}$ is expanding to the right but  according to the algorithm (line \ref{waitsToFinish}) $s_{k-1}$ waits to finish its expansion before passing on the expansion of $s_{k-2}$, so $s_{k-2}$ waits $O(n)$ time steps. Similarly each $s_{k-i}$ ($1 \leq i \leq k - 1$) must wait $O(n)$ time for $s_{k-i+1}$ to finish its expansion for a cost of $O(n)$. Moreover, each of these expansions is followed by a search and sampling that takes $O(n \log n)$ time; hence, time complexity of $\Theta(n^{3/2} \log^{3/2} n)$.

\begin{figure*}
\caption{Example of worst case time and message complexity}
\label{lollipop}
\centering
\includegraphics[width=.75\textwidth]{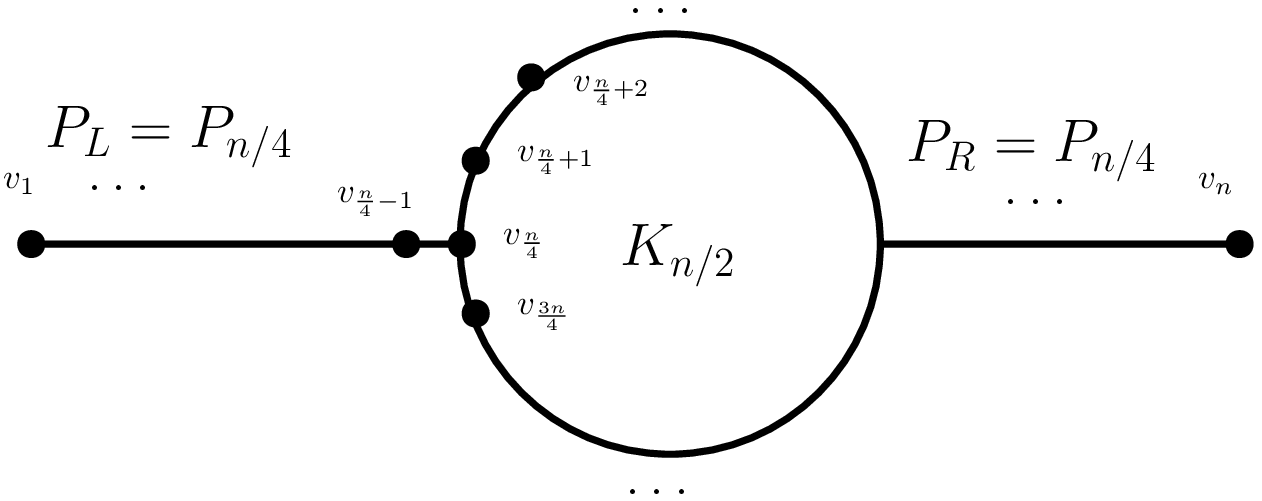}
\end{figure*}

\begin{algorithm}

\caption{Leader initiates Expand by sending $\langle Expand, tID \rangle $ to all of its children, and over edges in its $Found_L$ list.  Here, $x$ is any node.}
\label{expandMulti}

\begin{algorithmic}[1]

\Procedure{Expand-MultiLeader}{$tID$}
\State When node $x$ receives the message $\langle Expand, tID \rangle $ from node $t$ over an edge $e = (x, t)$:
\If{$tID < xID$} \texttt{//x responds immediately}
\State $x$ sends back $\langle \textit{Rejected-lower-ID}, tID \rangle$ over $e$.
\label{rejHigher}
\State If this is the first expansion message received from the fragment with identity $tID$, it adds $e$ to $Reject(x)$.
\Else
\label{waitsToFinish}
\Statex \qquad \parbox[t]{\dimexpr\linewidth-\algorithmicindent}{\texttt{//$x$ waits to finish its current expansion before handling these cases} \strut}
\If{$tID > xID$}
\label{compareID}
\State $x$ updates its fragment identity to $tID$ and the first node that $x$ receives $\langle Expand, tID \rangle $ from, becomes $x$'s parent in $T$. \texttt{//$x$ joins fragment $tID$}
\label{joinT}
\If{$x$ is part of some fragment $X$ upon receiving the expansion message}
\label{notFirstTime}
\State It forwards $\langle Expand, tID \rangle $ to its \textit{neighbors} in $X$ except the node $x$ has received the expansion message from, to the nodes in $Reject(x)$, and over edges in its Found lists.
\label{tochildrenless}
\Else
\label{isFirstTime}
\If{$x$ is a high-degree node \textbf{and} $x$ is not a star node}
\label{highNonStar}
\State If it has not received any $\langle \textit{Star} \rangle$ message yet, it waits to receive one, and \\ then adds the corresponding edge to $Found_O(x)$.
\State It forwards $\langle Expand, tID \rangle $ message over edges in $Found_L(x)$ and $Found_O(x)$.
\Else \textit{ }($x$ is a low-degree or star node)
\label{lowOrStar}
\State It forwards $\langle Expand, tID \rangle $ message to all of its neighbors.
\EndIf 
\EndIf 
\Else \textit{ } ($xID = tID$)
\State If the sender of $\langle Expand, tID \rangle $ is not $x$'s parent in $T$, it sends back $\langle \textit{Reject-same-tree}, tID \rangle$. Else, it only forwards $\langle Expand, tID \rangle $ to its children, nodes in $Reject(x)$, and over edges in its Found lists.
\label{tochildrenequal}
\EndIf
\State $x$ waits to hear back from all of the nodes it has forwarded to.
\State $x$ removes from the Found lists all the edges it has forwarded $\langle Expand, tID \rangle $ over.\label{emptyFoundlist}
\State $x$ removes from $Reject(x)$ any edge is has forwarded over.
\State $x$ updates its children to be all nodes that it has received $\langle Accept, tID \rangle $ from, in this Expand.
\label{updateBranches}
\State If $x$ is a leader, and receives a $\langle \textit{Rejected-lower-ID}, tID \rangle$, it stops running the FindST-Leader.
\label{stopProto}
\State If $x$ has received at least one $\langle \textit{Rejected-lower-ID}, tID \rangle$, it sends $\langle \textit{Rejected-lower-ID}, tID \rangle$ to its parent. Else, $x$ sends up $\langle Accept, tID \rangle$.
\label{rejConvergecast}
\EndIf
\EndProcedure
\end{algorithmic}
\end{algorithm}

\section{Conclusion}
We presented the first asynchronous algorithm for computing the MST in the CONGEST model with $\tilde{O}(n^{3/2})$ communication when nodes have initial knowledge of their neighbors' identities. This shows that the KT1 model  is significantly more communication efficient than KT0 even in the  asynchronous model. Open problems that are raised by these results are: (1) Does the asynchronous KT1 model require substantially more communication that the synchronous KT1 model?  (2) Can we improve the time complexity of the algorithm while maintaining the message complexity? 

\newpage
\bibliographystyle{spmpsci}      

%
%


\bibliography{mybib}

\begin{thebibliography}{10}
\providecommand{\url}[1]{{#1}}
\providecommand{\urlprefix}{URL }
\expandafter\ifx\csname urlstyle\endcsname\relax
  \providecommand{\doi}[1]{DOI~\discretionary{}{}{}#1}\else
  \providecommand{\doi}{DOI~\discretionary{}{}{}\begingroup
  \urlstyle{rm}\Url}\fi

\bibitem{ahn2012graph}
Ahn, K.J., Guha, S., McGregor, A.: Graph sketches: sparsification, spanners,
  and subgraphs.
\newblock In: Proceedings of the 31st ACM SIGMOD-SIGACT-SIGAI symposium on
  Principles of Database Systems, pp. 5--14. ACM (2012)

\bibitem{awerbuch1985complexity}
Awerbuch, B.: Complexity of network synchronization.
\newblock Journal of the ACM (JACM) \textbf{32}(4), 804--823 (1985)

\bibitem{awerbuch1987optimal}
Awerbuch, B.: Optimal distributed algorithms for minimum weight spanning tree,
  counting, leader election, and related problems.
\newblock In: Proceedings of the nineteenth annual ACM symposium on Theory of
  computing, pp. 230--240. ACM (1987)

\bibitem{awerbuch1990trade}
Awerbuch, B., Goldreich, O., Vainish, R., Peleg, D.: A trade-off between
  information and communication in broadcast protocols.
\newblock Journal of the ACM (JACM) \textbf{37}(2), 238--256 (1990)

\bibitem{awerbuch2007time}
Awerbuch, B., Kutten, S., Mansour, Y., Patt-Shamir, B., Varghese, G.: A
  time-optimal self-stabilizing synchronizer using a phase clock.
\newblock IEEE Transactions on Dependable and Secure Computing \textbf{4}(3)
  (2007)

\bibitem{awerbuch1990network}
Awerbuch, B., Peleg, D.: Network synchronization with polylogarithmic overhead.
\newblock In: Foundations of Computer Science, 1990. Proceedings., 31st Annual
  Symposium on, pp. 514--522. IEEE (1990)

\bibitem{elkin2004faster}
Elkin, M.: A faster distributed protocol for constructing a minimum spanning
  tree.
\newblock In: Proceedings of the fifteenth annual ACM-SIAM symposium on
  Discrete algorithms, pp. 359--368. Society for Industrial and Applied
  Mathematics (2004)

\bibitem{elkin2006unconditional}
Elkin, M.: An unconditional lower bound on the time-approximation trade-off for
  the distributed minimum spanning tree problem.
\newblock SIAM Journal on Computing \textbf{36}(2), 433--456 (2006)

\bibitem{elkin2008synchronizers}
Elkin, M.: Synchronizers, spanners.
\newblock In: Encyclopedia of Algorithms, pp. 1--99. Springer (2008)

\bibitem{Elkin:2017:DES:3055399.3055452}
Elkin, M.: Distributed exact shortest paths in sublinear time.
\newblock In: Proceedings of the 49th Annual ACM SIGACT Symposium on Theory of
  Computing, STOC 2017, pp. 757--770. ACM, New York, NY, USA (2017).
\newblock \doi{10.1145/3055399.3055452}.
\newblock \urlprefix\url{http://doi.acm.org/10.1145/3055399.3055452}

\bibitem{elkin2017simple}
Elkin, M.: A simple deterministic distributed mst algorithm, with near-optimal
  time and message complexities.
\newblock arXiv preprint arXiv:1703.02411  (2017)

\bibitem{emek2010efficient}
Emek, Y., Korman, A.: Efficient threshold detection in a distributed
  environment.
\newblock In: Proceedings of the 29th ACM SIGACT-SIGOPS symposium on Principles
  of distributed computing, pp. 183--191. ACM (2010)

\bibitem{faloutsos1995optimal}
Faloutsos, M., Molle, M.: Optimal distributed algorithm for minimum spanning
  trees revisited.
\newblock In: Proceedings of the fourteenth annual ACM symposium on Principles
  of distributed computing, pp. 231--237. ACM (1995)

\bibitem{gallager1983distributed}
Gallager, R.G., Humblet, P.A., Spira, P.M.: A distributed algorithm for
  minimum-weight spanning trees.
\newblock ACM Transactions on Programming Languages and systems (TOPLAS)
  \textbf{5}(1), 66--77 (1983)

\bibitem{garay1998sublinear}
Garay, J.A., Kutten, S., Peleg, D.: A sublinear time distributed algorithm for
  minimum-weight spanning trees.
\newblock SIAM Journal on Computing \textbf{27}(1), 302--316 (1998)

\bibitem{kapron2013dynamic}
Kapron, B.M., King, V., Mountjoy, B.: Dynamic graph connectivity in
  polylogarithmic worst case time.
\newblock In: Proceedings of the twenty-fourth annual ACM-SIAM symposium on
  Discrete algorithms, pp. 1131--1142. Society for Industrial and Applied
  Mathematics (2013)

\bibitem{Khan:2006:FDA:2136050.2136076}
Khan, M., Pandurangan, G.: A fast distributed approximation algorithm for
  minimum spanning trees.
\newblock In: Proceedings of the 20th International Conference on Distributed
  Computing, DISC'06, pp. 355--369. Springer-Verlag, Berlin, Heidelberg (2006)

\bibitem{king2015construction}
King, V., Kutten, S., Thorup, M.: Construction and impromptu repair of an mst
  in a distributed network with o (m) communication.
\newblock In: Proceedings of the 2015 ACM Symposium on Principles of
  Distributed Computing, pp. 71--80. ACM (2015)

\bibitem{kutten2015complexity}
Kutten, S., Pandurangan, G., Peleg, D., Robinson, P., Trehan, A.: On the
  complexity of leader election.
\newblock Journal of the ACM (JACM) \textbf{62}(1), 7 (2015)

\bibitem{kutten1995fast}
Kutten, S., Peleg, D.: Fast distributed construction of k-dominating sets and
  applications.
\newblock In: Proceedings of the fourteenth annual ACM symposium on Principles
  of distributed computing, pp. 238--251. ACM (1995)

\bibitem{mashreghi2017time}
Mashreghi, A., King, V.: Time-communication trade-offs for minimum spanning
  tree construction.
\newblock In: Proceedings of the 18th International Conference on Distributed
  Computing and Networking, p.~8. ACM (2017)

\bibitem{mashreghi2018broadcast}
Mashreghi, A., King, V.: Broadcast and minimum spanning tree with $ o (m) $
  messages in the asynchronous congest model.
\newblock arXiv preprint arXiv:1806.04328  (2018)

\bibitem{pandurangan2017time}
Pandurangan, G., Robinson, P., Scquizzato, M.: A time-and message-optimal
  distributed algorithm for minimum spanning trees.
\newblock In: Proceedings of the 49th Annual ACM SIGACT Symposium on Theory of
  Computing, pp. 743--756. ACM (2017)

\bibitem{peleg1987optimal}
Peleg, D., Ullman, J.D.: An optimal synchronizer for the hypercube.
\newblock In: Proceedings of the sixth annual ACM Symposium on Principles of
  distributed computing, pp. 77--85. ACM (1987)

\bibitem{sarma2012distributed}
Sarma, A.D., Holzer, S., Kor, L., Korman, A., Nanongkai, D., Pandurangan, G.,
  Peleg, D., Wattenhofer, R.: Distributed verification and hardness of
  distributed approximation.
\newblock SIAM Journal on Computing \textbf{41}(5), 1235--1265 (2012)

\bibitem{singh1995highly}
Singh, G., Bernstein, A.J.: A highly asynchronous minimum spanning tree
  protocol.
\newblock Distributed Computing \textbf{8}(3), 151--161 (1995)

\end{thebibliography}

\end{document}